\documentclass[letterpaper,journal]{IEEEtran}

\usepackage{amsmath,amssymb,amsfonts}
\usepackage{bm}
\usepackage{cite}
\usepackage{graphicx}
\usepackage[caption=false,font=footnotesize]{subfig}
\usepackage{textcomp}
\usepackage{url}
\usepackage{booktabs}
\usepackage{array}
\usepackage{algorithm}
\usepackage{algorithmic}
\usepackage{xcolor}
\usepackage{colortbl}
\usepackage{tikz}
\usepackage[colorlinks,citecolor=blue,urlcolor=blue,linkcolor=blue]{hyperref}

\newtheorem{theorem}{Theorem}

\newcommand{\R}{\mathbb{R}}
\newcommand{\T}{^{\mathsf T}}
\newcommand{\dd}{\mathrm d}
\newcommand{\norm}[1]{\left\lVert #1 \right\rVert}
\newcommand{\abs}[1]{\left\lvert #1 \right\rvert}
\newcommand{\atan}{\operatorname{atan}}
\newcommand{\atanTwo}{\operatorname{atan2}}

\definecolor{lime}{HTML}{A6CE39}
\definecolor{ourmethodblue}{RGB}{226,239,255}
\DeclareRobustCommand{\orcidicon}{
\begin{tikzpicture}
\draw[lime, fill=lime] (0,0)
circle[radius=0.16]
node[white]{{\fontfamily{qag}\selectfont \tiny \.{I}D}}; 
\end{tikzpicture}
\hspace{-2mm}
}
\foreach \x in {A, ..., Z}{%
\expandafter\xdef\csname orcid\x\endcsname{\noexpand\href{https://orcid.org/\csname orcidauthor\x\endcsname}{\noexpand\orcidicon}}
}

\begin{document}

\title{Parameter-Robust Sensorless Control of IPMSM Drives With Adaptive Flux Observer}

\author{
        
Fobao Zhou, \emph{Student Member, IEEE}, Zhenxiao Yin, \emph{Student Member, IEEE}, Xueyan Wang,
\\ Yang Shen, \emph{Student Member, IEEE}, Yuanfeng Qu, and Hang Zhao$^\ast$\hspace{-1.5mm}\orcidF{}, \emph{Member, IEEE}

\thanks{This work is supported by the National Natural Science Foundation of China (No. 52407066), Guangdong Science and Technology Program (2025A0505020029), and Youth S\&T Talent Support Program of Guangdong Provincial Association for Science and Technology (SKXRC2025464). (\textit{Corresponding author: Hang Zhao.)}}
\thanks{The authors are with the Robotics and Autonomous Systems Thrust, The
Hong Kong University of Science and Technology (Guangzhou), Guangzhou 511453, China
(e-mail: hangzhao@hkust-gz.edu.cn).}}

\markboth{JOURNAL OF \LaTeX\ CLASS FILES,~Vol.~xx, No.~xx, 2026}%
{Zhou \MakeLowercase{\textit{et al.}}: Parameter-Free Sensorless Control of IPMSM Drives}

\maketitle

\begin{abstract}
To address parameter sensitivity commonly in interior permanent magnet synchronous
motor (IPMSM) sensorless control, this paper proposes a parameter-robust
control framework by extending the adaptive flux observer from surface-mounted
PMSM to salient-pole machines. First, parameter
mismatches are interpreted as an equivalent flux
vector with $d$-axis and $q$-axis components. Permanent magnet flux,
$d$-axis inductance, and resistance mismatches are
mapped to the $d$-axis flux component, whereas $q$-axis inductance mismatch
is mapped to the $q$-axis component. Then, a scalar adaptive flux
update is designed to compensate the $d$-axis flux variation by tracking the
equivalent flux magnitude. Since the $q$-axis component rotates the observed
flux direction and cannot be eliminated by scalar adaptation, a high-frequency
$q$-axis voltage injection method is introduced to identify the $q$-axis
inductance. Meanwhile, stability analysis proves bounded equivalent flux
magnitude tracking and identifies the residual tangential flux responsible for
position error. Finally, experimental results show that the proposed method does not
require precise resistance, $d$-axis inductance, or flux, while the $q$-axis
inductance is provided by online identification.
\end{abstract}


\section{Introduction}
\IEEEPARstart{I}{nterior} permanent magnet synchronous motor (IPMSM) drives are
widely used in transportation, servo systems, robotics, and actuators because
of their high torque density, efficiency, and fast current response. In
field-oriented control, rotor electrical position is required to align stator
current with the rotor magnetic field. Replacing a mechanical position sensor
with a sensorless algorithm reduces cost and improves compactness and
reliability, but makes the controller dependent on the electromagnetic model
reconstructed from voltage and current measurements.

IPMSM sensorless techniques are generally divided into high-frequency signal
injection (HFI) and fundamental frequency model-based methods
\cite{WangReview2020}. HFI exploits rotor saliency and is mainly used in zero- and low-speed. Fundamental frequency methods estimate the back electromotive force (EMF) or
flux from measured voltages and currents, and are therefore preferred in
the medium- and high-speed regions. Representative implementations include
Kalman filter \cite{11197560}, extended state observer (ESO) \cite{Zuo2025ESO}, sliding-mode observer
(SMO) \cite{Wu2024FOSMO}, and nonlinear flux observer \cite{Lee2010}.

The accuracy of a fundamental frequency observer usually depends on the motor
parameters used in its model. The resistance varies with winding
temperature, the permanent magnet flux changes with temperature and aging, and
the $d$- and $q$-axis inductances vary with magnetic saturation,
cross-saturation, and load current. These variations distort the estimated
back-EMF or flux vector and consequently produce position error
\cite{Zhou2023Robust,Woldegiorgis2023,Wang2025ESOFlux}. Online parameter
identification can reduce this sensitivity, but simultaneous estimation of
multiple electrical parameters is constrained by identifiability and generally
requires additional estimators and excitation conditions
\cite{Liu2022SOESO,Woldegiorgis2023}. Therefore, improving robustness through
the observer structure remains an important solution.

Recently, several model-free and parameter-robust sensorless studies have been
reported for surface-mounted PMSM (SPMSM), reducing commissioning
dependence by reconstructing the electrical behavior from online data or an
ultralocal model \cite{Davari2026,Chen2025}. However, parameter-robust sensorless control for
IPMSMs remains more difficult, because saliency introduces unequal $d$- and
$q$-axis inductances whose mismatches affect the position-dependent flux vector
in different directions rather than as a single scalar inductance error.

Existing parameter-robust IPMSM designs still mainly
rely on SMO and ESO. SMO-based controllers improve robustness by using switching correction to suppress bounded
model uncertainties and disturbances. In \cite{Yin2022FSMO}, a full-order SMO
with a variable gain was used to reduce chattering and harmonic position
errors. A super-twisting flux SMO was proposed in \cite{Wu2025STFSMO} to
improve finite-time convergence, and adaptive or high-order SMO structures were
further combined with phase-locked loop or quadrature-signal generators in
\cite{Wang2025SuperTwisting,Wu2024FOSMO}. However, SMO
methods still face the usual tradeoff between chattering suppression, phase
delay, and gain tuning.

ESO-based methods take another route by treating unknown model terms and load
disturbances as extended states for online estimation and compensation. A
linear ESO was used in \cite{Qu2020ELADRC} to estimate the back EMF and
internal disturbances, while generalized-integrator and third-order ESOs were
combined in \cite{Zuo2025ESO} for back-EMF and position-disturbance
estimation. An adaptive harmonic-filtering ESO was also introduced in an
improved ADRC framework to suppress parameter mismatch and current harmonics
\cite{Yang2026ILADRC}. Related ESO-based active-flux
compensation and SMO--ESO schemes further improved robustness against
parameter and load disturbances \cite{Wang2025ESOFlux,Sun2024Robust}.
Nevertheless, these methods usually introduce several observer bandwidths,
filter parameters, or compensation gains. More importantly, parameter errors
are commonly lumped into a disturbance state, rather than separated according
to their physical effects on flux magnitude and direction.

To address the parameter sensitivity of IPMSM sensorless control, this paper
aims to develop a structurally simple and easily tuned sensorless algorithm.
The starting point of this work is the nonlinear flux observer for SPMSM
\cite{Lee2010}, which estimates rotor position through a geometric flux circle
constraint. A related flux observer was extended to IPMSM using a
position-dependent inductance matrix \cite{Khlaief2012}. In our previous
SPMSM work, the parameter mismatch problem was addressed by mapping
resistance, inductance, and flux errors to the flux side and compensating them
through an adaptive equivalent flux update \cite{zhou2026equivalent}. This method
removed the steady-state effects of resistance and flux mismatches, but
remained sensitive to inductance mismatch. Since resistance and flux mismatch
had been handled in the SPMSM case, a natural question is whether this
adaptive flux idea can be extended to IPMSMs for broader parameter mismatch
compensation and simultaneous inductance mismatch mitigation.

This paper answers this question by extending the adaptive flux principle to
IPMSM and establishing a parameter mismatch decomposition specific to
salient-pole machines. The
equivalent flux is separated into a normal component along the rotor $d$-axis
and a tangential component along the $q$-axis. Permanent magnet flux,
$d$-axis inductance, and the dominant resistance mismatches enter the normal
component and can therefore be compensated by the scalar adaptive flux update.
In contrast, $q$-axis inductance mismatch, together with a smaller
resistance-induced term when $i_d\neq0$, enters the tangential component and
rotates the estimated flux direction. A simple high-frequency $q$-axis voltage
injection method is consequently used to identify $L_q$. The identified value
is updated at no or light load and frozen under load, preventing
saturation-induced inductance drift from being introduced into the observer. The main contributions of this paper are summarized as follows.
\begin{itemize}
\item A unified normal--tangential equivalent flux decomposition is derived for
IPMSM parameter mismatch. It shows that flux, $L_d$, and the dominant
resistance mismatches mainly change the normal flux component, whereas $L_q$
mismatch directly produces a tangential position error.
\item An adaptive flux observer is developed to compensate the
normal mismatch component, while high-frequency $q$-axis voltage injection is
used to identify the remaining sensitive parameter $L_q$. Together, the two
mechanisms yield parameter-robust sensorless operation without precise offline
motor parameters.
\item A uniformly ultimately bounded analysis proves equivalent flux magnitude
tracking and separates the scalar tracking error from the residual tangential
flux. This result explains both the compensation capability of adaptive flux
and the necessity of $L_q$ identification.
\item Experiments under load transitions, sudden and time-varying parameter
mismatches, and a wide speed range verify the estimation accuracy, parameter
robustness, and convergence of the proposed method.
\end{itemize}


\section{IPMSM Model and Nonlinear Flux Observer}

The IPMSM model is established and analyzed in the stationary $\alpha\beta$
frame. Let
$i=[i_\alpha,i_\beta]\T$ and
$v=[v_\alpha,v_\beta]\T$ denote the stator current and voltage vectors,
respectively. The rotor electrical position is denoted by
$\theta$, and two unit vectors are introduced as:
\begin{equation}
n(\theta)=
\begin{bmatrix}
\cos\theta\\
\sin\theta
\end{bmatrix},
\quad
q(\theta)=
\begin{bmatrix}
-\sin\theta\\
\cos\theta
\end{bmatrix}
\label{eq:nq_def}
\end{equation}
where $n(\theta)$ and $q(\theta)$ are the rotor $d$- and $q$-axis unit vectors.

From \eqref{eq:nq_def}, it follows that $n\T n=q\T q=1$, $n\T q=0$, $\dot n=\omega q$, and
$\dot q=-\omega n$, where $\omega=\dot\theta$ is the electrical angular
speed. The current can be decomposed as:
\begin{equation}
i=i_d n+i_q q,
\quad
i_d=n\T i,
\quad
i_q=q\T i 
\label{eq:current_dec}
\end{equation}
where $i_d$ and $i_q$ are the $d$- and $q$-axis currents, respectively.

For an IPMSM, the stator flux vector is:
\begin{equation}
x=\lambda_{\alpha\beta}
=\left(\psi_f+L_d i_d\right)n+L_q i_q q
\label{eq:ipmsm_flux_dq}
\end{equation}
where $x=\lambda_{\alpha\beta}=[\lambda_{\alpha},\lambda_{\beta}]\T$ is the flux vector in the stationary frame, $L_d$
and $L_q$ are the real $d$- and $q$-axis inductances, and $\psi_f$ is the real flux linkage. Equivalently,
\begin{equation}
x=L(\theta)i+\psi_f n
\label{eq:ipmsm_flux_matrix}
\end{equation}
where $L(\theta)$ is the position-dependent inductance matrix
defined by $L(\theta)=L_d nn\T+L_q qq\T$.

Expanding this inductance matrix yields:
\begin{equation}
L(\theta)=
\begin{bmatrix}
L_0+L_1\cos 2\theta & L_1\sin 2\theta\\
L_1\sin 2\theta & L_0-L_1\cos 2\theta
\end{bmatrix}
\label{eq:L_theta}
\end{equation}
where $L_0=(L_d+L_q)/2$ is the average inductance and
$L_1=(L_d-L_q)/2$ is the saliency dependent inductance amplitude.

The voltage equation yields the flux model:
\begin{equation}
\dot x=v-R_s i
\label{eq:flux_state}
\end{equation}
where $R_s$ is the stator resistance. If the motor parameters are known,
an IPMSM nonlinear flux observer can be written as:
\begin{equation}
\dot{\hat x}
=v-R_s i+\Gamma \hat\eta\left(\psi_f^2-\norm{\hat\eta}^2\right)
\label{eq:ideal_nfo}
\end{equation}
where $\hat x$ is the estimate of $x$ and $\Gamma$ is the observer gain.
The auxiliary flux vector $\hat\eta$ is defined as:
\begin{equation}
\hat\eta=\hat x-L(\hat\theta)i 
\label{eq:ideal_eta}
\end{equation}
where $\hat\theta$ is the estimated rotor electrical position. The position is
obtained from the auxiliary flux direction:
\begin{equation}
\hat\theta=\atanTwo(\hat\eta_\beta,\hat\eta_\alpha)
\label{eq:theta_hat}
\end{equation}
where $\hat\eta_\alpha$ and $\hat\eta_\beta$ are the two components of
$\hat\eta$, and $\atanTwo(\cdot,\cdot)$ is the four-quadrant inverse tangent.
Compared with the SPMSM observer, the essential difference is that the
inductive current term is no longer $L_s i$, where $L_s$ is the scalar SPMSM
stator inductance, but
$L(\hat\theta)i$. Therefore, the observer couples the estimated position with
both $L_d$ and $L_q$.

\section{Adaptive Flux Nonlinear Observer}

\subsection{Unified Equivalent Flux Decomposition}

The first step is to derive the equivalent flux vector observed under parameter
mismatch. Let $\hat R_s$, $\hat L_d$, $\hat L_q$, and $\hat\psi_f$ denote the
nominal resistance, $d$-axis inductance, $q$-axis inductance, and flux used in the controller. The mismatches are defined as the real
values minus their nominal counterparts:
\begin{equation}
\begin{aligned}
\Delta R_s&=R_s-\hat R_s, &
\Delta L_d&=L_d-\hat L_d\\
\Delta L_q&=L_q-\hat L_q, &
\Delta\psi_f&=\psi_f-\hat\psi_f 
\end{aligned}
\label{eq:mismatch_def}
\end{equation}
where $\Delta R_s$, $\Delta L_d$, $\Delta L_q$, and $\Delta\psi_f$ denote the
stator resistance, $d$-axis inductance, $q$-axis inductance, and flux
mismatches, respectively. Resistance mismatch enters the flux integrator when
the nominal resistance is used:
\begin{equation}
v-\hat R_s i=\dot x+\Delta R_s i 
\label{eq:input_mismatch}
\end{equation}

To obtain a quasi-steady decomposition, the fast nonlinear correction is
neglected and the flux bias caused by resistance mismatch is written as:
\begin{equation}
\Delta x_R = \int \Delta R_s i\,\dd t
\label{eq:delta_x_R_def}
\end{equation}
where $\Delta x_R$ is an equivalent flux bias rather than a new physical motor
flux. Using $i=i_d n+i_q q$, $\int n\,\dd t= -q/\omega$, and
$\int q\,\dd t= n/\omega$, one obtains:
\begin{equation}
\Delta x_R
=
\frac{\Delta R_s}{\omega}i_q n
-
\frac{\Delta R_s}{\omega}i_d q 
\label{eq:delta_x_R}
\end{equation}

The equivalent flux vector $\eta_{\mathrm{eq}}$ generated by parameter mismatch is defined as:
\begin{equation}
\eta_{\mathrm{eq}}
= x+\Delta x_R-\hat L(\theta)i 
\label{eq:eta_mismatch_start}
\end{equation}
where
$\hat L(\theta)=\hat L_d nn\T+\hat L_q qq\T$ is the nominal inductance matrix
evaluated at $\theta$ for the local mismatch decomposition.

Substituting \eqref{eq:ipmsm_flux_dq} and \eqref{eq:delta_x_R} into
\eqref{eq:eta_mismatch_start} yields:
\begin{equation}
\eta_{\mathrm{eq}}
=
\psi_{\mathrm{eq},n} n
+
\psi_{\mathrm{eq},q} q
\label{eq:eta_eq_dec}
\end{equation}
where the normal equivalent flux $\psi_{\mathrm{eq},n}$ component is:
\begin{equation}
\psi_{\mathrm{eq},n}
=\hat\psi_f+\Delta\psi_f+\Delta L_d i_d
+\frac{\Delta R_s}{\omega}i_q
\label{eq:psi_eq_n}
\end{equation}
and the tangential equivalent flux $\psi_{\mathrm{eq},q}$ component is:
\begin{equation}
\psi_{\mathrm{eq},q}
=\Delta L_q i_q
-\frac{\Delta R_s}{\omega}i_d 
\label{eq:psi_eq_q}
\end{equation}

Equation \eqref{eq:eta_eq_dec} gives the equivalent flux vector under parameter
mismatch, while \eqref{eq:psi_eq_n} and \eqref{eq:psi_eq_q} are its projections
onto the rotor $d$- and $q$-axis unit vectors:
\begin{equation}
\psi_{\mathrm{eq},n}=n\T\eta_{\mathrm{eq}},
\quad
\psi_{\mathrm{eq},q}=q\T\eta_{\mathrm{eq}} 
\label{eq:eq_projection}
\end{equation}

Equivalently, these two projections can be collected as the equivalent flux
coordinate vector:
\begin{equation}
\begin{aligned}
{\psi}_{\mathrm{eq}}
&=
\begin{bmatrix}
\psi_{\mathrm{eq},n} & \psi_{\mathrm{eq},q}
\end{bmatrix}\T\\
\Psi&=\norm{{\psi}_{\mathrm{eq}}}
=\sqrt{\psi_{\mathrm{eq},n}^2+\psi_{\mathrm{eq},q}^2}
\end{aligned}
\label{eq:psi_eq_vector}
\end{equation}
where ${\psi}_{\mathrm{eq}}$ is the $nq$-coordinate representation of
$\eta_{\mathrm{eq}}$, and $\Psi$ is its equivalent flux magnitude.

The subscripts $n$ and $q$ denote the normal and tangential projections along
$n(\theta)$ and $q(\theta)$, respectively. Thus, flux, $L_d$, and resistance
mismatches affect the equivalent flux length through $\psi_{\mathrm{eq},n}$,
whereas $L_q$ mismatch and the $i_d$-dependent part of resistance mismatch
rotate its direction through $\psi_{\mathrm{eq},q}$. This resistance-induced
rotation appears under MTPA operation with $i_d\neq0$, but disappears under
$i_d=0$ control.

\subsection{Adaptive Flux Observer Design}

The baseline observer in \eqref{eq:ideal_nfo} constrains $\norm{\hat\eta}$ to a
fixed flux radius. The above analysis shows that parameter mismatch produces the
theoretical equivalent vector $\eta_{\mathrm{eq}}$ in
\eqref{eq:eta_eq_dec}. However, $\eta_{\mathrm{eq}}$ cannot be evaluated online
because the real parameters and their mismatches are unknown. It is therefore
used only to analyze the mismatch mechanism, not as an observer input.

Instead, let $\hat x_{\mathrm{pm}}$ denote the parameter-robust observer state
and let $\eta_{\mathrm{pm}}$ denote its online computable auxiliary flux:
\begin{equation}
\begin{aligned}
\eta_{\mathrm{pm}}
&=\hat x_{\mathrm{pm}}-\hat L(\hat\theta_{\mathrm{pm}})i \\
\hat\theta_{\mathrm{pm}}
&=\atanTwo(\eta_{\mathrm{pm},\beta},\eta_{\mathrm{pm},\alpha})
\end{aligned}
\label{eq:adaptive_eta}
\end{equation}
where $\hat\theta_{\mathrm{pm}}$ is the estimated rotor position,
$\eta_{\mathrm{pm},\alpha}$ and $\eta_{\mathrm{pm},\beta}$ are the
stationary frame components of $\eta_{\mathrm{pm}}$, and
$\hat L(\hat\theta_{\mathrm{pm}})$ has the form of \eqref{eq:L_theta} with
$\hat L_d$ and $\hat L_q$.

Let $\Psi>0$ denote the adaptive equivalent flux magnitude. To match both the
direction and magnitude of the online auxiliary flux, define the loss function:
\begin{equation}
\begin{aligned}
\mathcal T
&=\frac{1}{2}\left(\norm{\eta_{\mathrm{pm}}}-\Psi\right)^2
\end{aligned}
\label{eq:eq_flux_loss}
\end{equation}
where $\mathcal T\geq0$ is the equivalent flux estimation loss. Gradient
descent on \eqref{eq:eq_flux_loss} gives:
\begin{equation}
\dot{\Psi}
=-k_\psi\frac{\partial \mathcal T}{\partial \Psi}
=k_\psi\left(\norm{\eta_{\mathrm{pm}}}-\Psi\right)
\label{eq:eq_flux_gradient}
\end{equation}
where $k_\psi>0$ is the update gain and $\dot{\Psi}$ denotes the gradient update rate. To bound the update rate during large transients,
\eqref{eq:eq_flux_gradient} is implemented in the nonlinear form:
\begin{equation}
\dot{\Psi}
=k_\psi\tanh\!\left[
k_\psi\left(\norm{\eta_{\mathrm{pm}}}-\Psi\right)
\right]
\label{eq:eq_flux_update}
\end{equation}

The same gain is used outside and inside $\tanh(\cdot)$ to reduce the number of
tuning parameters. The update uses only the computable magnitude
$\norm{\eta_{\mathrm{pm}}}$, so the unknown theoretical vector
$\eta_{\mathrm{eq}}$ is not required by the implementation.

The adaptive equivalent flux magnitude is embedded in the nonlinear observer as:
\begin{equation}
\dot{\hat x}_{\mathrm{pm}}
=v-\hat R_s i+\Gamma_{\mathrm{pm}}\eta_{\mathrm{pm}}
\left(\Psi^2-\norm{\eta_{\mathrm{pm}}}^2\right)
\label{eq:adaptive_observer}
\end{equation}
where $\Gamma_{\mathrm{pm}}>0$ is the observer gain.

The two vectors have different roles. The vector $\eta_{\mathrm{eq}}$ is the
unknown quasi-steady mismatch equivalent derived in the decomposition above, whereas
$\eta_{\mathrm{pm}}$ is calculated at every sampling instant. Their difference,
using \eqref{eq:eta_mismatch_start} and \eqref{eq:adaptive_eta}, is:
\begin{equation}
        \eta_{\mathrm{pm}}-\eta_{\mathrm{eq}}
=\left(\hat x_{\mathrm{pm}}-x-\Delta x_R\right)
-\left[\hat L(\hat\theta_{\mathrm{pm}})-\hat L(\theta)\right]i 
\label{eq:eta_pm_eq_error}
\end{equation}

The objective of \eqref{eq:adaptive_eta}-\eqref{eq:adaptive_observer} is
therefore to make the computable vector $\eta_{\mathrm{pm}}$ track the unknown
equivalent flux $\eta_{\mathrm{eq}}$ as closely as possible. As indicated by
\eqref{eq:eta_pm_eq_error}, this tracking is achieved by adjusting
$\hat x_{\mathrm{pm}}$ so that the mismatch effects of flux,
stator resistance, and $d$- and $q$-axis inductances are absorbed on the
flux side.

\subsection{Compensation Mechanism for Parameter Mismatches}

The update law changes only the scalar radius associated with
$\eta_{\mathrm{pm}}$. It can therefore compensate mismatch terms that appear as
equivalent flux radius variations, but it cannot independently cancel a
tangential component. This distinction gives the four compensation mechanisms
below.

\textit{1) Flux mismatch:} When only $\Delta\psi_f$ exists,
\eqref{eq:psi_eq_n} and \eqref{eq:psi_eq_q} reduce to
$\psi_{\mathrm{eq},n}=\hat\psi_f+\Delta\psi_f=\psi_f$ and
$\psi_{\mathrm{eq},q}=0$. The mismatch only changes the radius of the
equivalent flux circle. Therefore, $\Psi$ follows
$\norm{\eta_{\mathrm{pm}}}=\psi_f$ without creating a steady angular
bias.

\textit{2) Resistance mismatch:} When only $\Delta R_s$ exists, the
normal and tangential parts become:
\begin{equation}
\psi_{\mathrm{eq},n}=\psi_f+\frac{\Delta R_s}{\omega}i_q,
\quad
\psi_{\mathrm{eq},q}=-\frac{\Delta R_s}{\omega}i_d 
\label{eq:R_mismatch_terms}
\end{equation}

The normal term $(\Delta R_s/\omega)i_q$ changes the equivalent flux radius
and is therefore represented in the adaptive radius. A tangential resistance
projection remains when $i_d\neq0$, but this residual is proportional to $i_d$ and
inversely proportional to the electrical speed $\omega$. Therefore, in the
medium- and high-speed region, resistance mismatch is largely compensated. The dominant tangential error source is
still the $L_q$ term $\Delta L_q i_q$, whereas the tangential resistance
projection depends jointly on speed and $d$-axis current.

\textit{3) $d$-axis inductance mismatch:} When only $\Delta L_d$ exists,
$\psi_{\mathrm{eq},n}=\psi_f+\Delta L_d i_d$ and
$\psi_{\mathrm{eq},q}=0$. The $L_d$ error is collinear with the
$d$-axis unit vector $n$, so it appears as a radius variation rather than a
direction error. This radial error can be compensated by the adaptive law
\eqref{eq:eq_flux_update}.

\textit{4) $q$-axis inductance mismatch:} When only $\Delta L_q$ exists, the
theoretical mismatch vector is
$\eta_{\mathrm{eq}}=\psi_f n+\Delta L_q i_q q$, which determines the
direction approached by $\eta_{\mathrm{pm}}$. The term
$\Delta L_q i_q q$ is tangential and tilts the observed vector away from the
true rotor $d$ axis. The resulting position error satisfies:
\begin{equation}
\Delta\theta_{L_q}
=\atan\left(\frac{\Delta L_q i_q}{\psi_f}\right)
\approx
\frac{\Delta L_q i_q}{\psi_f}
\label{eq:Lq_angle_error}
\end{equation}
where \(\Delta\theta_{L_q}\) is the additional position bias produced by
\(L_q\) mismatch, and the approximation holds for small errors. The adaptive law can only make
$\Psi$ approach $\norm{\eta_{\mathrm{pm}}}$, whose quasi-steady value is
$\sqrt{\psi_f^2+(\Delta L_q i_q)^2}$; it cannot remove
the orthogonal component $\Delta L_q i_q q$. Therefore, the effect of $L_q$
mismatch grows with load current and directly produces a error that
cannot be removed by the scalar flux update alone.

\begin{figure*}[!t]
    \centering
    \subfloat[]{
        \includegraphics[width=0.46\linewidth]{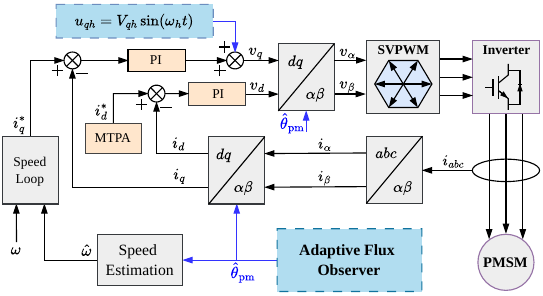}
    }
    \subfloat[]{
        \includegraphics[width=0.5\linewidth]{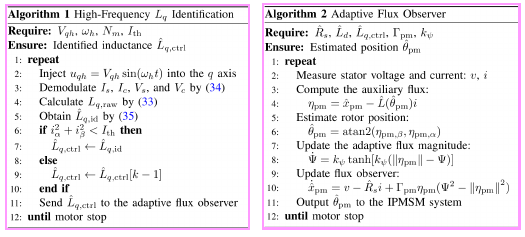}
    }
   \caption{\textcolor{black}{Proposed sensorless control framework for the
   IPMSM. (a) Overall block diagram. (b) Proposed adaptive flux observer
   and $L_q$ identification algorithms.} }
    \label{Structure diagram}
\end{figure*}

\section{Inductance Identification Strategy}

The previous section shows that the adaptive flux observer does not require
precise $\psi_f$, $L_d$, or $R_s$, but it remains sensitive to $L_q$ mismatch.
This section identifies $L_q$ using a small high-frequency
voltage injected into the $q$ axis.

\subsection{High-Frequency Small Signal Model}

Let $\omega_h=2\pi f_h$ denote the injection frequency. When
$\omega_h$ is much higher than the fundamental electrical frequency and the
bandwidth of the current controller, the fundamental current, speed, and
saturation operating point can be treated as slow variables within one
injection period. The high-frequency $q$-axis voltage and current satisfy the
local small signal model:
\begin{equation}
u_{qh}=V_{qh}\sin(\omega_h t)
=R_{\mathrm{hf}} i_{qh}+L_q\frac{\dd i_{qh}}{\dd t}
\label{eq:hf_rl_model}
\end{equation}
where $u_{qh}$ and $i_{qh}$ are the high-frequency $q$-axis voltage and current
components, $V_{qh}$ is the injected voltage amplitude, $R_{\mathrm{hf}}$ is
the effective resistance at the injection frequency, and $L_q$ denotes the
local effective $q$-axis inductance used by the observer. After synchronous
demodulation, write:
\begin{equation}
u_{qh}=V_s\sin\omega_h t+V_c\cos\omega_h t
\label{eq:hf_u}
\end{equation}
\begin{equation}
i_{qh}=I_s\sin\omega_h t+I_c\cos\omega_h t 
\label{eq:hf_i}
\end{equation}
where $V_s$ and $V_c$ are the sine and cosine coefficients of $u_{qh}$, and
$I_s$ and $I_c$ are the sine and cosine coefficients of $i_{qh}$.
Substituting \eqref{eq:hf_i} into \eqref{eq:hf_rl_model} and matching sine and
cosine coefficients gives:
\begin{equation}
V_s=R_{\mathrm{hf}}I_s-\omega_h L_q I_c,
\quad
V_c=R_{\mathrm{hf}}I_c+\omega_h L_q I_s 
\label{eq:hf_coeff}
\end{equation}

The antisymmetric combination cancels $R_{\mathrm{hf}}$:
\begin{equation}
I_sV_c-I_cV_s
=\omega_h L_q\left(I_s^2+I_c^2\right)
\label{eq:hf_anti}
\end{equation}

Therefore, the raw $q$-axis inductance estimate $L_{q,\mathrm{raw}}$ is:
\begin{equation}
L_{q,\mathrm{raw}}
=
\frac{I_sV_c-I_cV_s}
{\omega_h\left(I_s^2+I_c^2\right)}
\label{eq:Lq_raw}
\end{equation}

\subsection{Synchronous Demodulation and Filtering}

For the $m$-th identification window $\mathcal W_m$ containing an integer
number of injection periods, the demodulated coefficients are computed in the
hardware system as:
\begin{equation}
\begin{aligned}
I_s(m)&=\frac{2}{N_m}\sum_{k\in\mathcal W_m}i_q[k]\sin_h[k]\\
I_c(m)&=\frac{2}{N_m}\sum_{k\in\mathcal W_m}i_q[k]\cos_h[k]\\
V_s(m)&=\frac{2}{N_m}\sum_{k\in\mathcal W_m}u_{qh}[k]\sin_h[k]\\
V_c(m)&=\frac{2}{N_m}\sum_{k\in\mathcal W_m}u_{qh}[k]\cos_h[k]
\end{aligned}
\label{eq:demod}
\end{equation}
where $m$ is the window index, $N_m$ is the number of samples in
$\mathcal W_m$, and $\sin_h[k]$ and $\cos_h[k]$ are the sampled synchronous
references at the injection frequency. The raw estimate is then processed by a three-point
median filter and a low-pass filter:
\begin{equation}
\hat L_{q,\mathrm{id}}(m)
=\operatorname{LPF}\{\operatorname{MED}[L_{q,\mathrm{raw}}(m)]\}.
\label{eq:Lq_median}
\end{equation}
where $\operatorname{MED}[\cdot]$ denotes a median operation,
$\operatorname{LPF}\{\cdot\}$ denotes a low-pass filter, and
$\hat L_{q,\mathrm{id}}$ is the filtered $q$-axis inductance sent to the
controller once per demodulation window.

Noted that the experiments in \cite{Khlaief2012} show that a decrease in $L_q$ can
significantly degrade the accuracy of the estimated position. For an IPMSM,
magnetic saturation under load may also make the identified incremental
inductance decrease. If this load-dependent value is directly sent to the
observer, the tangential flux error may increase instead of being reduced.
Therefore, only the inductance identified near no load or light load is used
for observer calibration, which is implemented by a current threshold. With
$k$ is the sampling instant and $I_{\mathrm{th}}>0$ denoting the current
threshold, the actual input $\hat L_{q,\mathrm{ctrl}}$ is:
\begin{equation}
\hat L_{q,\mathrm{ctrl}}[k]
=
\begin{cases}
\hat L_{q,\mathrm{id}}[k],
& i_\alpha^2[k]+i_\beta^2[k]<I_{\mathrm{th}}\\
\hat L_{q,\mathrm{ctrl}}[k-1],
& i_\alpha^2[k]+i_\beta^2[k]\geq I_{\mathrm{th}}
\end{cases}
\label{eq:Lq_current_gate}
\end{equation}

This gate freezes the observer input under load and prevents saturation-induced
inductance drift from being injected continuously into the tangential flux
channel. With this identification loop, the dominant tangential mismatch
becomes:
\begin{equation}
\Delta L_q i_q
=\left(L_q-\hat L_{q,\mathrm{ctrl}}\right)i_q 
\label{eq:Lq_after_id}
\end{equation}
where $L_q-\hat L_{q,\mathrm{ctrl}}$ is the difference between the real and
identified $q$-axis inductances. If the current-gated identification reduces
this difference compared with the fixed nominal value, the tangential flux and
position-error bounds are reduced. The overall IPMSM control framework is shown in
Fig. \ref{Structure diagram}.

\section{Stability Analysis}

The nonlinear update law \eqref{eq:eq_flux_update} is driven by the online
auxiliary flux \(\eta_{\mathrm{pm}}\). Define its magnitude as:
\begin{equation}
s(t)=\norm{\eta_{\mathrm{pm}}(t)}
\label{eq:s_pm_def}
\end{equation}
where \(s(t)\) is the equivalent flux magnitude available to the adaptive law.
This definition does not require the unknown theoretical vector
\(\eta_{\mathrm{eq}}\) in \eqref{eq:eta_eq_dec}. 

The analysis is carried out on a compact operating set where the online flux
magnitude and its rate are bounded, namely
$0<s_{\min}\leq s(t)\leq s_{\max}$ and
$\abs{\dot s(t)}\leq\Phi$, where \(s_{\min}\), \(s_{\max}\), and \(\Phi\)
are positive constants.

\subsection{Equivalent Flux Magnitude Tracking}

Define the equivalent flux magnitude tracking error as:
\begin{equation}
\tilde{\Psi}
=\Psi-s(t)
=\Psi-\norm{\eta_{\mathrm{pm}}}
\label{eq:psi_eq_error}
\end{equation}

\begin{theorem}
\label{thm:eq_flux_mag_tracking}
For the IPMSM system described by \eqref{eq:ipmsm_flux_dq} and
\eqref{eq:input_mismatch}, with the observer \eqref{eq:adaptive_eta},
\eqref{eq:adaptive_observer} and the equivalent flux update law
\eqref{eq:eq_flux_update} under parameter mismatches, the magnitude \(\Psi\) is uniformly bounded. The tracking error
\(\tilde{\Psi}\) in \eqref{eq:psi_eq_error} converges to the compact set
\(\Omega_\Psi\), defined as follows:
\begin{equation}
\Omega_\Psi=
\left\{
\tilde{\Psi}\in\R\;\big|\;
|{\tilde{\Psi}}|
\leq
\epsilon_{\mathrm{eq}}
\right\}, \ \epsilon_{\mathrm{eq}}=\frac{\Phi}{k_\psi(k_\psi-\Phi)}
\label{eq:psi_eq_error_bound}
\end{equation}
\end{theorem}

\begin{IEEEproof}
From \eqref{eq:eq_flux_update} and \eqref{eq:psi_eq_error}, we have:
\begin{equation}
\dot{\tilde{\Psi}}
=-k_\psi\tanh(k_\psi\tilde{\Psi})-\dot s 
\label{eq:psi_eq_error_dot}
\end{equation}

Choose the Lyapunov function $V_\Psi$ as follows:
\begin{equation}
V_\Psi=\frac{1}{2}\tilde{\Psi}^2 
\label{eq:V_psi}
\end{equation}
Its derivative satisfies:
\begin{equation}
\begin{aligned}
\dot V_\Psi
&=\tilde{\Psi}\dot{\tilde{\Psi}} \leq
-k_\psi|{\tilde{\Psi}}|
 |{\tanh(k_\psi\tilde{\Psi})}|
+|{\tilde{\Psi}}|\Phi 
\end{aligned}
\label{eq:V_psi_dot}
\end{equation}

Using
\(\abs{\tanh a}\geq\abs{a}/(1+\abs{a})\), define
\(\gamma_\psi=|{\tilde{\Psi}}|\). For
\(\gamma_\psi>0\), \eqref{eq:V_psi_dot} gives:
\begin{equation}
\dot\gamma_\psi
\leq
-\frac{k_\psi^2\gamma_\psi}{1+k_\psi\gamma_\psi}
+\Phi 
\label{eq:gamma_psi_compare}
\end{equation}

The right-hand side of \eqref{eq:gamma_psi_compare} is strictly decreasing and
has the unique positive equilibrium \(\epsilon_{\mathrm{eq}}\) defined in
\eqref{eq:psi_eq_error_bound}. The comparison principle implies that
\(\gamma_\psi\) ultimately enters and remains in
\([0,\epsilon_{\mathrm{eq}}]\). Therefore, \eqref{eq:psi_eq_error_bound} holds.
\end{IEEEproof}

Theorem~\ref{thm:eq_flux_mag_tracking}
proves that the update law makes
\(\Psi\) track the online equivalent flux magnitude
\(\norm{\eta_{\mathrm{pm}}}\). For a constant operating point,
\(\Phi=0\), and \(\tilde{\Psi}\) converges to zero.

\subsection{Normal and Tangential Tracking Error Bounds}

The scalar result above does not imply that the tangential flux has been
eliminated. To project the result without identifying
\(\eta_{\mathrm{pm}}\) with the theoretical \(\eta_{\mathrm{eq}}\), decompose
the online auxiliary flux in the true rotor frame as:
\begin{equation}
\begin{aligned}
\eta_{\mathrm{pm}}
&=\psi_{\mathrm{pm},n}n+\psi_{\mathrm{pm},q}q\\
\psi_{\mathrm{pm},n}&=n\T\eta_{\mathrm{pm}},\quad
\psi_{\mathrm{pm},q}=q\T\eta_{\mathrm{pm}}
\end{aligned}
\label{eq:eta_pm_nq}
\end{equation}
where \(\psi_{\mathrm{pm},n}\) and \(\psi_{\mathrm{pm},q}\) are the online
normal and tangential flux components. Consequently,
\begin{equation}
s
=\sqrt{\psi_{\mathrm{pm},n}^2+\psi_{\mathrm{pm},q}^2},
\quad
n_{\mathrm{pm}}
=\frac{\psi_{\mathrm{pm},n}}{s}n
+\frac{\psi_{\mathrm{pm},q}}{s}q 
\label{eq:n_pm_nq}
\end{equation}
where \(n_{\mathrm{pm}}=\eta_{\mathrm{pm}}/\norm{\eta_{\mathrm{pm}}}\) is the
unit vector along the online auxiliary flux.
On the same compact operating set, define the bounded projection ratios:
\begin{equation}
\rho_n=\sup_t\frac{\abs{\psi_{\mathrm{pm},n}(t)}}{s(t)},
\quad
\rho_q=\sup_t\frac{\abs{\psi_{\mathrm{pm},q}(t)}}{s(t)} 
\label{eq:projection_ratios}
\end{equation}
where \(\rho_n\) and \(\rho_q\) denote the maximum fractions of the online
auxiliary flux magnitude projected onto the normal and tangential directions,
respectively.

Define the reconstructed equivalent flux vector as
\(\eta_{\mathrm{ad}}=\Psi n_{\mathrm{pm}}\). Relative to the \(d\)-axis
component \(\psi_{\mathrm{pm},n}n\), define the $d$-axis and tangential
residuals as
\(\varepsilon_d=n\T(\eta_{\mathrm{ad}}-\psi_{\mathrm{pm},n}n)\) and
\(\varepsilon_q=q\T(\eta_{\mathrm{ad}}-\psi_{\mathrm{pm},n}n)\).

\begin{theorem}
 \label{thm:dq_residual_bounds}
Under the conditions of
Theorem~\ref{thm:eq_flux_mag_tracking}
and with the definitions in \eqref{eq:eta_pm_nq}--\eqref{eq:projection_ratios}, the
$d$-axis and tangential residuals converge to the compact sets \(\Omega_d\)
and \(\Omega_q\), defined as follows:
\begin{equation}
\begin{aligned}
\Omega_d
&=
\left\{
\varepsilon_d\in\R\;\big|\;
\abs{\varepsilon_d}
\leq
\epsilon_{\mathrm{eq}}\rho_n
\right\}\\
\Omega_q
&=
\left\{
\varepsilon_q\in\R\;\big|\;
\abs{\varepsilon_q}
\leq
\epsilon_{\mathrm{eq}}\rho_q + \sup_t\abs{\psi_{\mathrm{pm},q}(t)}
\right\}
\end{aligned}
\label{eq:dq_residual_bound}
\end{equation}
\end{theorem}

\begin{IEEEproof}
Substituting \eqref{eq:n_pm_nq} into
\(\eta_{\mathrm{ad}}=\Psi n_{\mathrm{pm}}\) gives:
\begin{equation}
\eta_{\mathrm{ad}}
=\Psi\frac{\psi_{\mathrm{pm},n}}{s}n
+\Psi\frac{\psi_{\mathrm{pm},q}}{s}q
\label{eq:eta_ad_expand}
\end{equation}

Because \(s^2=\psi_{\mathrm{pm},n}^2+\psi_{\mathrm{pm},q}^2\), subtracting the
$d$-axis component \(\psi_{\mathrm{pm},n}n\) yields:
\begin{equation}
\begin{aligned}
\eta_{\mathrm{ad}}-\psi_{\mathrm{pm},n}n
&=
\left(\Psi\frac{\psi_{\mathrm{pm},n}}{s}-\psi_{\mathrm{pm},n}\right)n
+\Psi\frac{\psi_{\mathrm{pm},q}}{s}q\\
&=
\tilde{\Psi}\frac{\psi_{\mathrm{pm},n}}{s}n
+\left(\psi_{\mathrm{pm},q}
+\tilde{\Psi}\frac{\psi_{\mathrm{pm},q}}{s}\right)q 
\end{aligned}
\label{eq:eta_ad_minus_d}
\end{equation}

Projecting
\eqref{eq:eta_ad_minus_d} onto the true basis vectors \(n\) and
\(q\) gives:
\begin{equation}
\varepsilon_d
=\tilde{\Psi}
\frac{\psi_{\mathrm{pm},n}}{s},
\quad
\varepsilon_q
=\psi_{\mathrm{pm},q}
+\tilde{\Psi}
\frac{\psi_{\mathrm{pm},q}}{s}
\label{eq:dq_residual}
\end{equation}

Combining \eqref{eq:dq_residual}, \eqref{eq:psi_eq_error_bound}, and
\eqref{eq:projection_ratios} proves \eqref{eq:dq_residual_bound}.
\end{IEEEproof}

Theorem~\ref{thm:dq_residual_bounds}
follows by projecting the online vector \(\eta_{\mathrm{pm}}\).
When \(\tilde{\Psi}\to0\), the $d$-axis residual tends to zero, whereas the
tangential residual tends to \(\psi_{\mathrm{pm},q}\). From
\eqref{eq:psi_eq_q}, the corresponding tangential channel is
\(\Delta L_q i_q-(\Delta R_s/\omega)i_d\). Thus, scalar equivalent flux
adaptation removes the magnitude error but cannot independently cancel the
tangential component dominated by \(L_q\) mismatch. Therefore, \(L_q\)
identification in Section IV is needed to reduce this remaining error.

\begin{figure}[!t]
    \centering
        \includegraphics[width=0.65\linewidth]{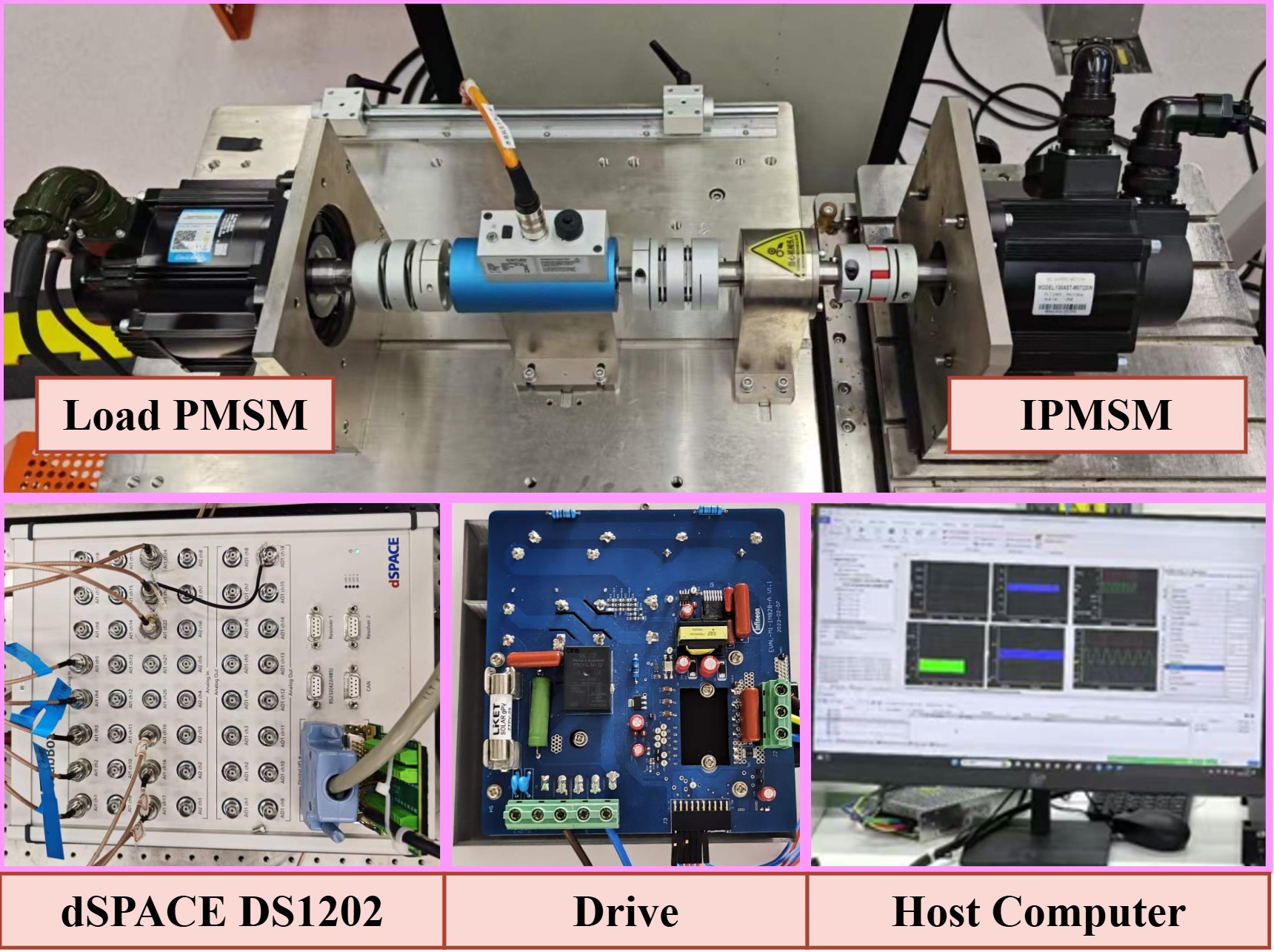} 
    \caption{ Experimental platform of IPMSM.} 
    \label{Motor Platform}
\end{figure}

\begin{table}[!t]
  \caption{Parameters of the IPMSM System}
  \renewcommand\arraystretch{1.25} 
  \small
  \centering
  \setlength{\tabcolsep}{5mm} 
  \begin{tabular}{c c c } \hline \hline
    Description & Parameter & Value   \\ \hline 
    Pole pairs & $N_p$ & 5  \\
    DC voltage (V) & $U_{dc}$ & 300 \\ 
    Rated power (kW) & $P_r$ & 1.5 \\
    Rated torque (N$\cdot$m) & $T_r$ & 7 \\
    Rated current (A) & $I_r$ & 6 \\
    \hline
    $d$-axis inductance (mH) & $L_d$ & 7.9  \\
    $q$-axis inductance (mH) & $L_q$ & 11.2  \\
    Stator resistance ($\Omega$) & $R_s$ & 0.495 \\
    Flux linkage (Wb) & $\psi_f $ & 0.117   \\ \hline
    Observer gain & $\Gamma_{\mathrm{pm}}$ & 5000  \\
    Update gain & $k_\psi$ & 10 \\
    \hline  \hline 
  \end{tabular}
  \label{Table.Parameters}
\end{table}

\section{Experimental Results and Analysis}

To validate the proposed parameter-robust IPMSM sensorless control scheme, an
experimental platform is built as shown in Fig. \ref{Motor Platform}. The setup consists of two
coupled PMSMs, a dSPACE DS1202 controller, a power driver, and a host computer,
allowing the adaptive flux observer and the \(L_q\) identification loop to be
tested under controlled load conditions. The speed loop is executed at 2 kHz,
whereas the current loop and PWM switching frequency are both set to 10 kHz.
The main IPMSM parameters are listed in Table \ref{Table.Parameters}. For
comparison, three sensorless controllers are used: Method 1 is the
flux observer in \cite{Khlaief2012}, Method 2 is the active-flux ESO-based
controller in \cite{Zhou2023Robust}, and Method 3 is the ADRC-based controller
in \cite{Yang2026ILADRC}. The proposed controller contains only two independent
tuning gains $\Gamma_{\mathrm{pm}}$ and $k_\psi$, which can be selected after a few practical tuning iterations.
All parameter mismatch comparison tests are carried out under full-load
operation.

\begin{figure}[!t]
    \centering
    \includegraphics[width=0.96\linewidth]{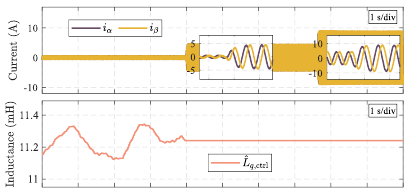}
    \vspace{-8pt}
    \caption{\textcolor{black}{Experimental results of $q$-axis inductance identification $\hat L_{q,\mathrm{ctrl}}$
    under no-load, half-load, and full-load conditions.}}
    \label{Lq_identificatio}
\end{figure}

\begin{figure}[!t]
    \centering
    \captionsetup[subfloat]{captionskip=0.05pt}
    \subfloat[]{
        \includegraphics[width=0.48\linewidth]{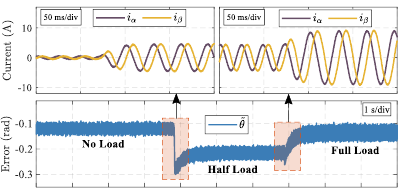}
    }
    \hspace{-7pt}
     \subfloat[]{
        \includegraphics[width=0.48\linewidth]{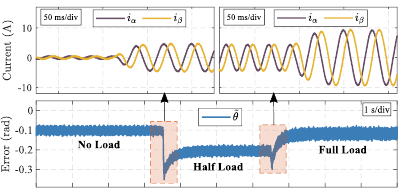}
    }\\ \vspace{-10pt}
    \subfloat[]{
        \includegraphics[width=0.48\linewidth]{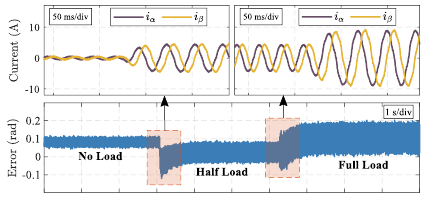}
    }
    \hspace{-7pt}
    \subfloat[]{
        \includegraphics[width=0.48\linewidth]{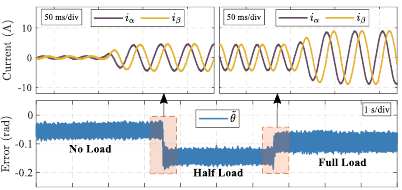}
    }
    \caption{\textcolor{black}{Experimental results under no-load, half-load, and full-load conditions at 500 rpm. (a) Method 1. (b) Method 2. (c) Method 3. (d) Our method.}}
    \label{half_load}
\end{figure}

\begin{figure}[!t]
    \centering
    \captionsetup[subfloat]{captionskip=0.05pt}
    \subfloat[]{
        \includegraphics[width=0.48\linewidth]{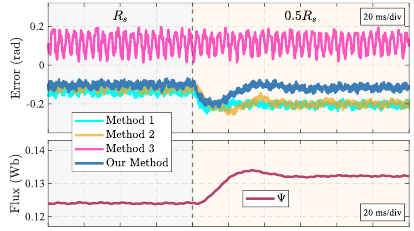}
    }
    \hspace{-7pt}
     \subfloat[]{
        \includegraphics[width=0.48\linewidth]{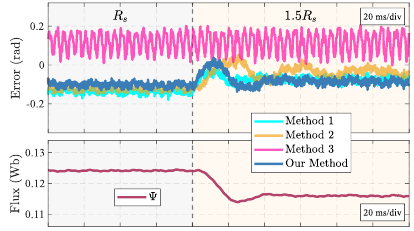}
    }
    \caption{\textcolor{black}{Experimental results under resistance mismatch at 500 rpm, showing the adaptive flux \(\Psi\) update. (a) \(R_s \to 0.5R_s\). (b) \(R_s \to 1.5R_s\).}}
    \label{Rs_0.5_1.5}
\end{figure}

\begin{figure}[!t]
    \centering
    \captionsetup[subfloat]{captionskip=0.05pt}
    \subfloat[]{
        \includegraphics[width=0.48\linewidth]{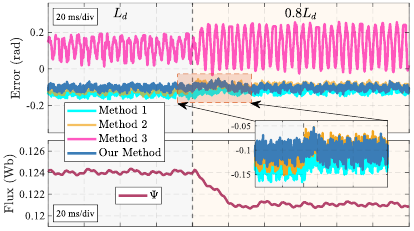}
    }
    \hspace{-7pt}
     \subfloat[]{
        \includegraphics[width=0.48\linewidth]{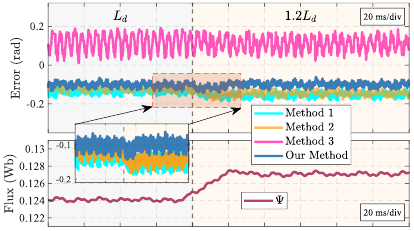}
    }
    \caption{\textcolor{black}{Experimental results under $d$-axis inductance mismatch at 500 rpm, showing the adaptive flux \(\Psi\) update. (a) \(L_d \to 0.8L_d\). (b) \(L_d \to 1.2L_d\).}}
    \label{Ld_0.8_1.2}
\end{figure}

\subsection{Performance of Inductance Identification}

The previous analysis has shown that the $L_q$ mismatch cannot be eliminated
by the adaptive flux update because it acts on the tangential equivalent flux
component. To achieve parameter-robust operation, a $q$-axis rotating voltage
injection is used to identify $L_q$, and the value obtained under no-load or
light-load conditions is then used in the proposed observer. In the experiment,
the injected voltage is $u_{qh}=3\sin(800\pi t)$, the sample number $N_m$ in
\eqref{eq:demod} is set to 20, and the current threshold $I_{\mathrm{th}}$ in \eqref{eq:Lq_current_gate} is
set to 0.5. The identification result is shown in Fig.
\ref{Lq_identificatio}. It can be observed that the identified $\hat L_{q,\mathrm{ctrl}}$ fluctuates
around the real $L_q$ value under no-load operation, which indicates accurate
inductance identification. When a larger load is applied, the identified value
is frozen.

\subsection{Performance of Load Transition}

Fig. \ref{half_load} compares the estimation performance of the four
methods at 500 rpm when the operating condition is suddenly
changed from no load to half load and then to full load. Overall, as the load
increases, the estimation errors of all methods first decrease
and then increase. Method 1 and Method 2 produce large transient errors during
load switching, with peak values of approximately 0.3-0.35 rad.
Method 3 reduces the transient error to some extent, but its error fluctuation
becomes clearly larger under full-load operation. In contrast, the proposed
method limits the transient error within approximately 0.2 rad
during both load transitions and quickly recovers to the steady-state error
range. 

\begin{table}[!t] \color{black}
  \caption{Root Mean Square Error (RMSE) of Estimation Errors under Different Parameter Mismatch Conditions. Unit: rad}
  \renewcommand\arraystretch{1.25} 
  \small
  \centering
  \setlength{\tabcolsep}{2mm} 
  \begin{tabular}{c c c c c} \hline \hline
     & Method 1 & Method 2 & Method 3 & \textbf{Our Method}   \\ \hline 
    Nominal & $0.136$ & $0.121$ & $0.118$ & {\boldmath$0.103$} \\ \hline
    $0.5R_s$ & $0.202$ & $0.194$ &$ 0.118$ & {\boldmath$0.106$} \\
    $1.5R_s$ & $0.076$ & ${0.039}$ & $0.119$ & {\boldmath$0.101$} \\ \hline
    $0.8L_d$  & $0.121$ & ${0.096}$ & $0.144$ & {\boldmath$0.104$}  \\
    $1.2L_d$  & $0.154$ & $0.144$ & $0.109$ & {\boldmath$0.103$}   \\ \hline
    $0.9L_q$  & $0.223$ & ${0.205}$ & ${0.012}$ & {\boldmath$0.186$}  \\
    $1.1L_q$  & $0.045$ & $0.028$ & $0.250$ & {\boldmath$0.104$}   \\ \hline
    $0.9\psi_f$ & $0.194$ & $0.206$ & $-$ & {\boldmath$0.103$}  \\
    $1.1\psi_f$ & $0.071$ & $0.025$ & ${-}$ & {\boldmath$0.103$}  \\
    \hline  \hline 
  \end{tabular}
  \label{Table.RMSE}
\end{table}

\subsection{Performance of $R_s$ and $L_d$ Mismatch}

Fig. \ref{Rs_0.5_1.5} shows the experimental results when \(R_s\) is suddenly
changed to 0.5\(R_s\) and 1.5\(R_s\), respectively. Method 1 and Method 2 show
clear steady-state shifts in the position estimation error, whereas Method 3
always exhibits large periodic error oscillations. In contrast, the proposed
method only produces a short transient disturbance at the switching instant
and then rapidly returns to the steady-state error range of approximately $-$0.1 rad.

The results under \(L_d\) mismatch are shown in Fig. \ref{Ld_0.8_1.2}. When
\(L_d\) is changed to 0.8\(L_d\) and 1.2\(L_d\), Method 1 and Method 2 exhibit
small steady-state error shifts, while Method 3 shows a larger error
oscillation under 0.8\(L_d\). The proposed method maintains a small and stable
error range in both cases and recovers quickly after the parameter changes.
These results confirm that the adaptive flux update can adjust the equivalent
flux according to the \(R_s\) and \(L_d\) mismatches, thereby improving
steady-state robustness and transient recovery.

\begin{figure}[!t]
    \centering
    \captionsetup[subfloat]{captionskip=0.05pt}
    \subfloat[]{
        \includegraphics[width=0.48\linewidth]{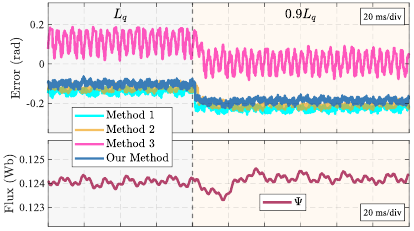}
    }
    \hspace{-7pt}
     \subfloat[]{
        \includegraphics[width=0.48\linewidth]{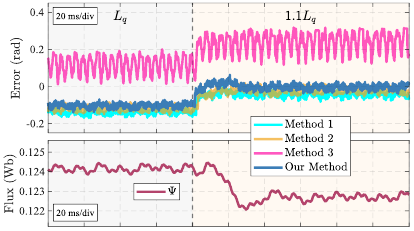}
    }
    \caption{\textcolor{black}{Experimental results under $q$-axis inductance mismatch at 500 rpm, showing the adaptive flux \(\Psi\) update. (a) \(L_q \to 0.9L_q\). (b) \(L_q \to 1.1L_q\).}}
    \label{Lq_0.9_1.1}
\end{figure}
\begin{figure}[!t]
    \centering
    \captionsetup[subfloat]{captionskip=0.05pt}
    \subfloat[]{
        \includegraphics[width=0.48\linewidth]{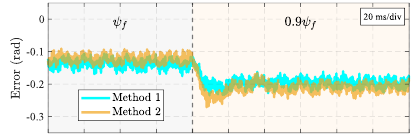}
    }
    \hspace{-7pt}
     \subfloat[]{
        \includegraphics[width=0.48\linewidth]{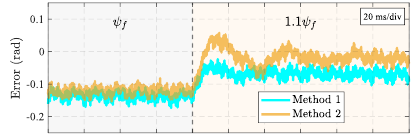}
    }\\ \vspace{-10pt}
    \subfloat[]{
        \includegraphics[width=0.48\linewidth]{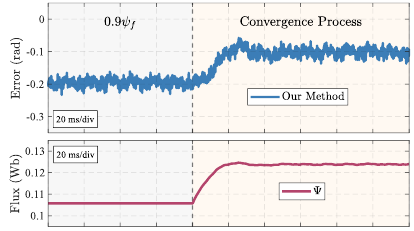}
    }
    \hspace{-7pt}
    \subfloat[]{
        \includegraphics[width=0.48\linewidth]{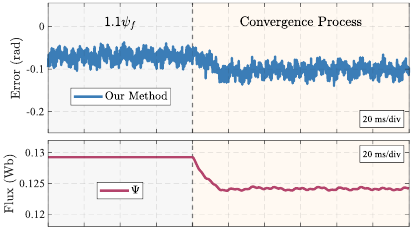}
    }
    \caption{\textcolor{black}{Experimental results under flux mismatch at
    500 rpm. (a) \(\psi_f \to 0.9\psi_f\). (b) \(\psi_f \to
    1.1\psi_f\). (c) Our method with initial \(0.9\psi_f\). (d) Our
    method with initial \(1.1\psi_f\). Method 3 does not require the flux
    parameter.}}
    \label{Psi_0.9_1.1}
\end{figure}

\subsection{Performance of $L_q$ and $\psi_f$ Mismatch}

Fig. \ref{Lq_0.9_1.1} shows the results when \(L_q\) is changed to
\(0.9L_q\) and \(1.1L_q\). The position estimation errors of all four methods
change noticeably: reducing \(L_q\) shifts the error in the negative
direction, whereas increasing \(L_q\) shifts it in the positive direction.
This result confirms that the adaptive flux update cannot fully compensate
\(L_q\) mismatch, so accurate \(L_q\) acquisition or online identification is
necessary.

Fig. \ref{Psi_0.9_1.1} further gives the results under flux mismatch. When the
flux is changed to \(0.9\psi_f\) and \(1.1\psi_f\), Method 1 and Method 2 show
clear transient errors and steady-state shifts, indicating their sensitivity
to the flux parameter. Method 3 is almost unaffected because its observer
model does not require \(\psi_f\). For the proposed method, the adaptive flux
converges to approximately \(0.124~\mathrm{Wb}\) from both initial values, and
the estimation error returns to a similar steady-state range,
verifying its compensation capability under flux mismatch. Table \ref{Table.RMSE}
summarizes the error variations of the four methods under four parameter
mismatches.

\subsection{Performance of Large Range Parameter Mismatch}

To further evaluate robustness under large parameter uncertainties, Figs.
\ref{Rs_0.2_1.8} and \ref{Ld_0.2_1.8} show the experimental results under
larger \(R_s\) and \(L_d\) mismatches, respectively. When \(R_s\) or \(L_d\) is
suddenly changed to \(0.2\) and \(1.8\) times the rated value, the adaptive
flux is adjusted accordingly. The position estimation error only exhibits a
limited transient deviation and rapidly returns to a stable range, while the
system remains stable throughout the whole test.

Fig. \ref{Rs_Ld_sine} further examines the performance under time-varying
parameters, where \(R_s\) and \(L_d\) vary sinusoidally with an amplitude of
80\% of their nominal values. Under sinusoidal \(L_d\) variation, the error
fluctuation remains almost unchanged. Under sinusoidal \(R_s\) variation, the
error shows a small slow fluctuation because, as analyzed above, part of the
resistance error is projected onto the tangential flux when \(i_d\neq0\). These
results confirm that the proposed observer can handle both sudden and
time-varying parameter
mismatches.

\begin{figure}[!t]
    \centering
    \captionsetup[subfloat]{captionskip=0.05pt}
    \subfloat[]{
        \includegraphics[width=0.96\linewidth]{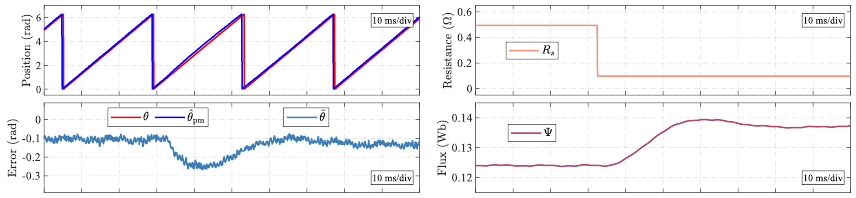}
    }\\ \vspace{-10pt}
    \subfloat[]{
        \includegraphics[width=0.96\linewidth]{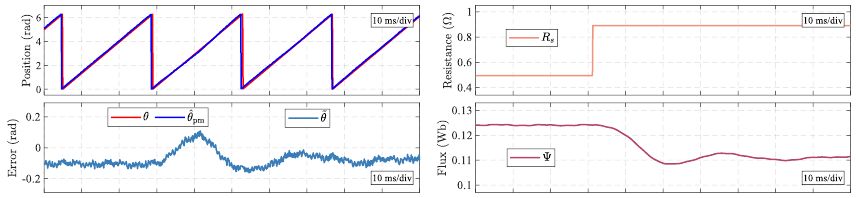}
    }
    \caption{\textcolor{black}{Experimental results of the adaptive flux
    observer under resistance mismatch at 500 rpm. (a)
    \(R_s \to 0.2R_s\). (b) \(R_s \to 1.8R_s\).}}
    \label{Rs_0.2_1.8}
\end{figure}

\begin{figure}[!t]
    \centering
    \captionsetup[subfloat]{captionskip=0.05pt}
    \subfloat[]{
        \includegraphics[width=0.96\linewidth]{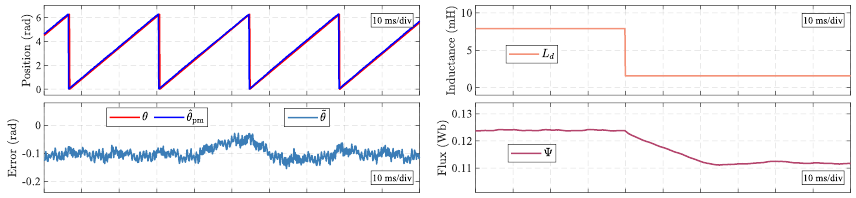}
    }\\ \vspace{-10pt}
    \subfloat[]{
        \includegraphics[width=0.96\linewidth]{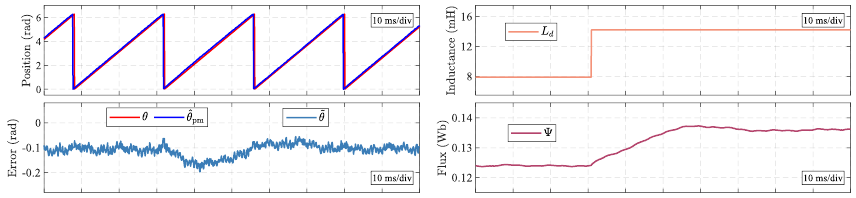}
    }
    \caption{\textcolor{black}{Experimental results of the adaptive flux
    observer under $d$-axis inductance mismatch at 500 rpm. (a)
    \(L_d \to 0.2L_d\). (b) \(L_d \to 1.8L_d\).}}
    \label{Ld_0.2_1.8}
\end{figure}

\begin{figure}[!t]
    \centering
    \captionsetup[subfloat]{captionskip=0.05pt}
    \subfloat[]{
        \includegraphics[width=0.48\linewidth]{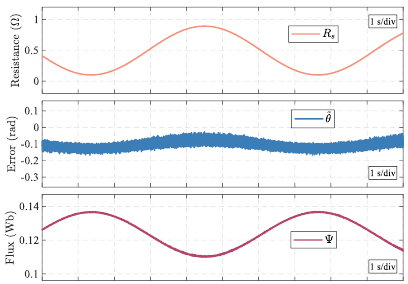}
    } \hspace{-7pt}
    \subfloat[]{
        \includegraphics[width=0.48\linewidth]{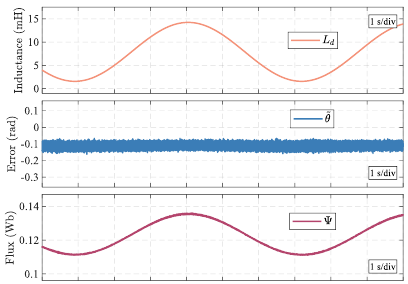}
    }
    \caption{\textcolor{black}{Experimental results of the adaptive flux
    observer under sinusoidal parameter variations at 500 rpm. (a)
    $R_s + 0.8R_s\sin(2\pi t)$. (b) $L_d + 0.8L_d\sin(2\pi t)$.}}
    \label{Rs_Ld_sine}
\end{figure}

\begin{figure}[!t]
    \centering
    \captionsetup[subfloat]{captionskip=0.05pt}
    \subfloat[]{
        \includegraphics[width=0.96\linewidth]{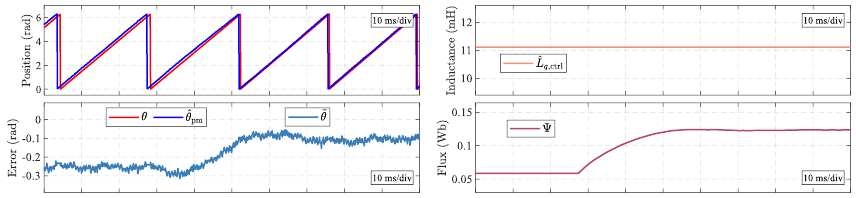}
    }\\ \vspace{-10pt}
    \subfloat[]{
        \includegraphics[width=0.96\linewidth]{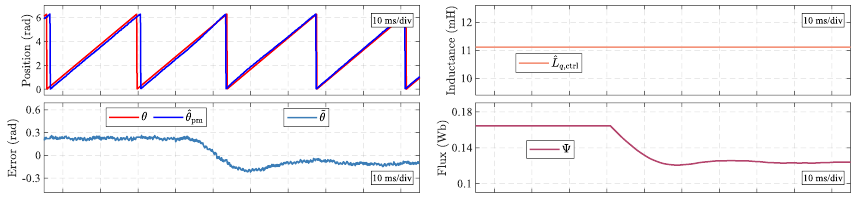}
    }
    \caption{\textcolor{black}{Experimental results of the adaptive flux
    observer with different initial flux values at 500 rpm. (a) Initial
    flux \(0.5\psi_f\). (b) Initial flux \(1.5\psi_f\).}}
    \label{flux_0.5_1.5_adaptive}
\end{figure}

\subsection{Performance of Flux and Identified Inductance Update}

Fig. \ref{flux_0.5_1.5_adaptive} shows the results with larger initial flux mismatches. When the initial values are set to \(0.5\psi_f\) and
\(1.5\psi_f\), the adaptive flux in both cases converges to approximately
\(0.124~\mathrm{Wb}\), and the position estimation error returns to a similar
steady-state range. This indicates good convergence of the adaptive update
under large initial flux mismatch.

Fig. \ref{Lq_0.8_1.2_adaptive} verifies the effectiveness of the identified
\(L_q\). After the controller parameter is switched from \(0.8L_q\) or
\(1.2L_q\) to \(\hat L_{q,\mathrm{ctrl}}\), the original estimation
bias is rapidly reduced and returns to the stable range of approximately $-$0.1 rad. Thus, the identified \(L_q\) effectively corrects the
inductance mismatch and improves the position estimation accuracy.


\begin{figure}[!t]
    \centering
    \captionsetup[subfloat]{captionskip=0.05pt}
    \subfloat[]{
        \includegraphics[width=0.96\linewidth]{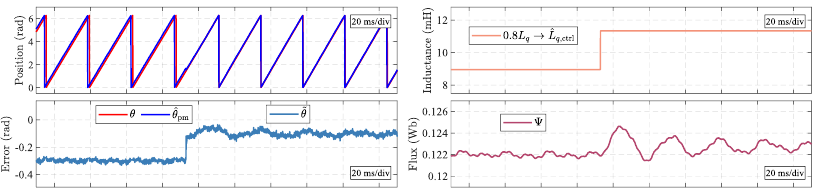}
    }\\ \vspace{-10pt}
    \subfloat[]{
        \includegraphics[width=0.96\linewidth]{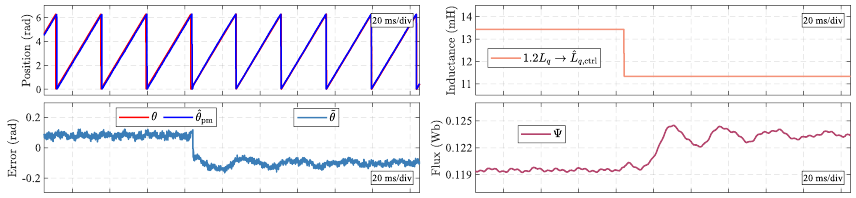}
    }
    \caption{\textcolor{black}{Experimental results when different
    \(q\)-axis inductances are switched to the identified value at 500 rpm.
    (a) $0.8 L_q \to \hat L_{q,\mathrm{ctrl}}$. (b) $1.2 L_q \to
    \hat L_{q,\mathrm{ctrl}}$.}}
    \label{Lq_0.8_1.2_adaptive}
\end{figure}

\subsection{Performance of Acceleration}

Fig. \ref{200-1000} shows the experimental results when the motor accelerates
from 200 rpm to 1000 rpm. The speed smoothly follows the acceleration command,
and the estimated position continuously tracks the measured position. The
position estimation error decreases from approximately $-$0.15 rad at low
speed and settles within a small range near 1000 r/min. Meanwhile, the
adaptive flux changes smoothly and finally converges, verifying the dynamic
tracking capability and stable sensorless operation over a wide speed range.

\begin{figure}[!t]
    \centering
    \includegraphics[width=0.8\linewidth]{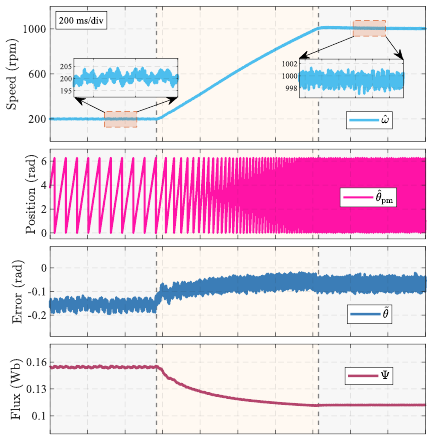}
    \vspace{-8pt}
    \caption{\textcolor{black}{Experimental results during acceleration from
    200 rpm to 1000 rpm.}}
    \label{200-1000}
\end{figure}

\section{Conclusion}

This paper has presented a parameter-robust IPMSM sensorless control method
based on an adaptive flux observer and high-frequency $q$-axis inductance
identification. By mapping parameter mismatches into the equivalent flux
domain, the adaptive flux update compensates the normal flux variations caused
by flux, $L_d$, and the dominant resistance errors, while $L_q$
identification reduces the remaining tangential sensitivity. The stability
analysis proves bounded equivalent flux magnitude tracking and clarifies the
remaining tangential residual. Compared with the benchmark methods, experiments
show that the proposed controller restores the estimation error to its nominal
level under resistance, $d$-axis inductance, $q$-axis inductance, and flux
mismatches.

\bibliographystyle{IEEEtran}
\bibliography{parameter_robust_ipmsm_refs}

\end{document}